\documentclass[11pt]{article}
\pdfoutput=1

\usepackage[utf8]{inputenc} 
\usepackage[T1]{fontenc}    
\usepackage{hyperref}       
\usepackage{url}            
\usepackage{booktabs}       
\usepackage{amsfonts}       
\usepackage{nicefrac}       
\usepackage{microtype}      

\usepackage{amstext}
\usepackage{amsmath}
\usepackage{amsthm}
\usepackage{amssymb}
\usepackage{bm}
\usepackage{wrapfig}
\usepackage{amsfonts, amsmath}
\usepackage[utf8]{inputenc}
\usepackage{graphicx}
\usepackage{bbm, comment}
\usepackage{subfigure}
\usepackage{color}
\usepackage{float}
\usepackage{wrapfig}
\usepackage{enumitem}
\usepackage{url}
\usepackage{MnSymbol,wasysym,marvosym,ifsym}

\usepackage[ruled]{algorithm2e} 

\makeatletter
\def\BState{\State\hskip-\ALG@thistlm}
\makeatother

\usepackage{scalerel,stackengine}
\stackMath
\newcommand\reallywidehat[1]{%
\savestack{\tmpbox}{\stretchto{%
  \scaleto{%
    \scalerel*[\widthof{\ensuremath{#1}}]{\kern-.6pt\bigwedge\kern-.6pt}%
    {\rule[-\textheight/2]{1ex}{\textheight}}
  }{\textheight}%
}{0.5ex}}%
\stackon[1pt]{#1}{\tmpbox}%
}
\usepackage{tabulary}
\usepackage{booktabs}

\usepackage[top=1.in, bottom=1.in, left=1.in, right=1.in]{geometry}
\usepackage[numbers]{natbib}

\newtheorem*{theorem*}{Theorem}
\newtheorem{theorem}{Theorem}[section]
\newtheorem{lemma}[theorem]{Lemma}

\newtheorem{claim}[theorem]{Claim}
\newtheorem{assumption}[theorem]{Assumption}

\newtheorem{definition}[theorem]{Definition}

\usepackage{xcolor}

\DeclareMathOperator{\enc}{Enc}
\DeclareMathOperator{\dec}{Dec}
\DeclareMathOperator{\view}{view}
\DeclareMathOperator{\simu}{sim}

\begin{document}
\title{Securely Trading Unverifiable Information without Trust}
\author{Yuqing Kong\\ Peking University \and Yiping Ma\\ Peking University \and Yifan Wu\\ Peking University}
\date{}
\maketitle

\begin{abstract}
In future, information may become one of the most important assets in economy.
 However, unlike common goods (e.g. clothing), information is troublesome in trading since the information commodities are \emph{vulnerable}, as they lose their values immediately after revelation, and possibly \emph{unverifiable}, as they can be subjective. By authorizing a trusted center (e.g. Amazon) to help manage the information trade, traders are ``forced'' to give the trusted center the ability to become an information monopolist. 

To this end, we need a \emph{trust-free} (i.e. without a trusted center and with only strategic traders) unverifiable information trade protocol such that it 1) motivates the sellers to provide high quality information, and the buyer to pay for the information with a fair price (\emph{truthful}); 2) except the owner, the information is known only to its buyer if the trade is executed (\emph{secure}). 

In an unverifiable information trade scenario (e.g. a medical company wants to buy experts' opinions on multiple difficult medical images with unknown pathological truth from several hospitals), we design a trust-free, truthful, and secure protocol, \emph{Smart Info-Dealer (SMind)} , for  information trading, by borrowing three cutting-edge tools that include peer prediction, secure multi-party computation, and smart contract. With SMind, without a trusted center, a seller with high-quality information is able to sell her information securely at a fair price and those with low-quality information cannot earn extra money with poor information or steal information from other sellers. We believe SMind will help describe a free and secure information trade scenario in future. 


\end{abstract}



\section{Introduction}\label{sec:intro}
Nowadays people come to realize the importance of information assets.
 However, information commodities differ from most entity goods substantially because \emph{information commodities lose their values immediately after revelation}. Customers can look around a clothing store before making their purchases, while an information store should hide all information from the public, otherwise everyone learns the information and no one will pay for it. However, without the chance to see the goods, how can a buyer decide whether to pay for certain pieces of information as the way she shops in the clothing store?

Moreover, even if a piece of information is disclosed, an information buyer would still be confused, especially when it is unverifiable (e.g. subjective opinions). A costumer shopping her clothes can make a purchase based on the style of the clothing, while the information buyer may have no clue about the value of the disclosed information, which can be meaningless or even fake. 

Furthermore, even if there exists a trusted center who is able to magically evaluate the quality of the unverifiable information without revealing it to the buyer, traders are ``forced'' to authorize this powerful center to access the information.
A commonly-trusted information trading center may gradually grow into an information monopolist, which has been observed in the cases of google, Amazon etc. Also, dependency on a trusted center severely limits the trade of the information asset.  

Despite several challenges as mentioned above, information trading is prevalent in contemporary world. A running example is that a medical company (the buyer) wants to buy labels (e.g. hard labels like benign/malignant, or soft labels like 90\% benign) for multiple difficult medical images with \emph{unknown pathological truth} from multiple hospitals (the sellers). To tackle the challenges discussed above, we aim to design an information trade protocol for \emph{unverifiable} information, which satisfies the following three properties:



\begin{description}
\item[Trust-free] a trusted center is not required, and every trader, including the hospitals and the medical company, is rational/strategic rather than required to be honest

\item [Truthful]  1) each hospital (the seller) is incentivized to provide truthful and high-quality information, i.e. labels for the medical images; 2) the medical company (the buyer) is incentivized to pay the information with a fixed payment function;

\item [Secure] except the final trade price, each hospital cannot infer additional information about other hospital's information. The medical company can only obtain the information if and only if the trade protocol is successfully executed.
\end{description}

A recent line of research, peer prediction (e.g.\cite{miller2005eliciting,prelec2004bayesian,dasgupta2013crowdsourced,2016arXiv160501021K,2016arXiv160303151S,kong2018water}), focuses on designing mechanisms that elicit truthful, but unverifiable information. Peer prediction mainly measures each seller's information quality based on her \emph{peer}'s reported information and pays for the measured quality. In our running example of two hospitals A and B, the quality of their information is measured according to the ``similarity'' between their labels. The ``similarity'' measure is carefully designed such that the hospitals will be incentivized to report honestly and a hospital who reports meaningless information (e.g. label randomly or always label benign) will obtain the lowest payment. However, peer prediction does not ensure privacy of participants and implicitly requires a trusted center to compute the ``similarity'' and to make the exchange of payment and information. 

A cryptography protocol, secure multi-party computation (MPC \cite{yao1986generate,goldreich1987play,beaver1990round, ben1988completeness})\footnote{MPC is the abbreviation of secure multi-party computation in standard cryptographic notion, we will use this notion in the subsequent part of this paper.}, allows us to address the above security issue. MPC is a multi-party protocol such that the process of computation reveals nothing beyond the output from the perspective of protocol participants. Thus, combined with peer prediction, MPC allows sellers to compute the quality of the information without violating privacy of each sellers. However, different from the traditional MPC scenario where a party can be either honest or malicious, in our setting all parties are strategic.


Although MPC achieves privacy-preserving computation of information quality, i.e. the payment of the information, the trust-free fair exchange of payment and information still has not been guaranteed. This gives rise to new incentive issues. A recent work \cite{dziembowski2018fairswap} proposes a smart contract based solution for the trust-free fair exchange. 

We borrow this smart contract based solution and combine with the other two cutting-edge tools, peer prediction, MPC, to propose a trust-free, truthful, and secure information trade protocol for unverifiable information, \emph{Smart Info-Dealer (SMind)}. In our running example with SMind which does not require the existence of trusted center, the medical company is able to buy high-quality information securely from high-quality hospitals at a fair price and a low-quality hospital cannot earn money with poor information or steal other hiqh-quality hospitals' information. 
Thus, in a new world where the information becomes one of the most important assets in economy, we believe our method will help describe a free and secure information trade scenario.

\subsection{Road map}
Section \ref{sec:related} discusses related work. Section \ref{sec:prelim} introduces basic game theory and cryptographic concepts and also the main three building blocks of SMind. Section \ref{sec:smartmpcpp} formally introduces the information trade setting, the protocol design goals and our protocol, SMind. Section \ref{sec:proof} shows our main theorem: SMind is trust-free, truthful, and secure. Section \ref{sec:ext} shows a natural extension from 2-seller setting to multiple sellers. In the end, Section \ref{sec:futurework} concludes by discussing the robustness and implementation of SMind.

\section{Related work}\label{sec:related}

\paragraph{Decentralized prediction market} With the bloom of researches on blockchain, decentralized information trading platforms are developed on chain \cite{peterson2018augur, team2017gnosis, adler2018astraea, yandeaux}, mostly addressing the decentralization of prediction market. However, their settings are totally different from our work. The prediction market is essentially an \emph{verifiable} information trade platform and assumes that the ground truth will be revealed in the future. Also, our work enables instant payment, while participants in prediction market should wait for the truth to be revealed. 
\paragraph{Peer prediction} A recent line of research, peer prediction (e.g.\cite{miller2005eliciting,prelec2004bayesian,dasgupta2013crowdsourced,2016arXiv160501021K,2016arXiv160303151S,kong2018water}), focuses on designing mechanisms that elicit truthful, but unverifiable information. Unlike prediction market, peer prediction does not assume the existence of ground truth. However, peer prediction does not consider the security issue and implicitly requires a trusted center to compute the payment and to transfer the payment and the information. Our work raises and addresses the security issues of peer prediction. 
\paragraph{Outsourced computation}Works on outsourced computation aim to verify the correctness of computation \cite{walfish2015verifying, lamport1982byzantine, dong2017betrayal}. However, outsource computation addresses mostly on verifiable computation, while our setting considers unverifiable information. Also, their works do not take into account security issues and assume that there is a trusted judge for arbitration, while ours ensures privacy and removes a trusted center from the trading. 
\paragraph{Secure Multi-party Computation}
The notion of secure multi-party computation is first introduced as an open problem by Yao in 1980s\cite{yao1986generate}, and later there comes many protocols such as \cite{goldreich1987play,beaver1990round, ben1988completeness}, etc. Since 2000s building practical systems using general-purpose multi-party computation becomes realistic due to algorithmic and computing improvements \cite{malkhi2004fairplay, damgaard2012multiparty, nielsen2012new, wang2017authenticated}. MPC has been applied to varies of areas, e.g. auction \cite{bogetoft2009secure, bogetoft2006practical}, machine learning \cite{mohassel2017secureml, mohassel2018aby} to address the security and privacy issue. In the current paper, we apply MPC to a new field, eliciting unverifiable information, to address the security issue. Works on \emph{strategic MPC} introduce the notion of rational protocol design in MPC \cite{garay2013rational, izmalkov2005rational, asharov2011towards}. Our protocol can also be considered under the strategic MPC setting. By carefully designing economic rewards and punishments, the parties will be motivated to run MPC truthfully. However, we consider an information trade setting which is totally different from previous strategic MPC works.
\paragraph{Blockchain and smart contract}
Firstly proposed in Bitcoin \cite{nakamoto2008bitcoin}, blockchain is built upon consensus protocols, with security guaranteed by honest majority on chain \cite{garay2015bitcoin, sompolinsky2015secure}.  Ethereum \cite{buterin2014next} first realizes smart contract on blockchain. It can be defined as enforced agreements between distrusted parties \cite{clack2016smart}. In recent works on fairly exchanging \emph{verifiable} digital goods \cite{teutsch2017scalable, dziembowski2018fairswap}, smart contract works as a trusted third person to verify the incorrectness of trade process when needed. We borrow this idea and combine with peer prediction, MPC to propose a truthful, secure, trust-free protocol for the trade of \emph{unverifiable} information.


\section{Preliminaries}\label{sec:prelim}

\subsection{Basic game theory topics: extensive form, subgame, strategy, equilibrium concepts}\label{sec:gametheory}
Readers can refer to \cite{osborne2004introduction} for a detailed definition of basic game theory concepts. 

A \emph{game} consists of
a list of players,
a description of the players' possible actions,
a specification of what the players know at their turn and
a specification of the payoffs of players' actions.

A node $x$ in the extensive-form game defines a \emph{subgame} if  $x$ and its successors are in an information set that does not contain nodes that are not successors of $x$. A subgame is the tree structure initiated by such a node $x$ and its successors.

A \emph{strategy} is a complete plan for a player in the game, which describes the actions that the player would take at each of her possible decision points.

We define a strategy profile $\mathbf{s}$ as a profile of all agents' strategies $(s_1,s_2,...,s_n)$. Agents play $\mathbf{s}$ if for all $i$, agent $i$ plays strategy $s_i$.

A {\em (Bayesian) Nash equilibrium (B(NE))} consists of a strategy profile $\mathbf{s} = (s_1, \ldots, s_n)$ such that no agent wishes to change her strategy since other strategies will decrease her expected payment, given the strategies of the other agents and the information contained in her prior and her signal.

A strategy profile is a \emph{strong Nash equilibrium} if it represents a Nash equilibrium in which no coalition, taking the actions of its complements as given, can cooperatively deviate in a way that benefits all of its members.

A strategy profile is a \emph{subgame perfect equilibrium (SPE)} if it represents a Nash equilibrium of every subgame of the original game.

\begin{definition}[Strong Subgame Perfect Equilibrium (Strong SPE)] \label{def:strongspe}
A strategy profile is a \emph{strong subgame perfect equilibrium} if it is a SPE and a strong NE.
\end{definition}

\emph{Backward induction procedure} is the process of analyzing a game from the end to the beginning. At each decision node, one strikes from consideration any actions that are dominated, given the terminal nodes that can be reached through the play of the actions identified at successor nodes.



\subsection{Cryptographic Building Blocks}\label{sec:crypt}
\subsubsection{Basic cryptographic tools}
\begin{definition}[Encryption Scheme]
    A encryption scheme is a tuple $(\text{Gen}, \enc, \dec)$ that
    \begin{description}
        \item\textit{key generation}: having a security parameter $n$, generates a key $k\leftarrow \text{Gen}(1^{\kappa})$
        \item\textit{encryption}: Upon input a message $m$ and key $k$, output a ciphertext $ c\leftarrow \enc_k(m)$
        \item\textit{decryption}: Upon input a ciphertext $c$ and key $k$, outputs a message $m\leftarrow \dec_k(c)$
    \end{description}
\end{definition}

\begin{definition}[Commitment Scheme]
    A commitment scheme is a two-party protocol between sender and a receiver. In the first phase, the sender \textit{commits} to some value $m$ and sends the commitment to the receiver; in the second phase, the sender can \textit{open} the commitment by revealing $m$ and some auxiliary information to the receiver, which the receiver will use to verify that the value he received is indeed the value the sender committed to during the first phase.\par
    A commitment scheme is often defined as follows:\par
    \begin{description}[ font=\normalfont]\label{des:commit}
        \item\textit{Commit}: To \textit{commit to} a message $m$, the sender choose randomness \texttt{Op}\footnote{The detail of this step depends on  concrete construction, we only provide high-level description here.}, and generates the \textit{commitment} by $\texttt{Com}(m)\leftarrow Commit(\texttt{Op}, m)$.
        The sender sends \texttt{Com} to receiver.
         \item\textit{Open}: To open a \textit{commitment}, the sender sends $m$ and \texttt{Op} to receiver. Receiver computes $checkbit \leftarrow Open(\texttt{Com},\texttt{Op},m)$. If $checkbit$ is $valid$, then the receiver accepts that this message $m$ is indeed the message that the sender previously committed to. Otherwise she rejects.
    \end{description}

    A commitment scheme should satisfy the following two properties:\par
    \begin{description}\label{des:commitproperty}
        \item\textit{Hiding}: Given the commitment $\texttt{Com}(m)$ of a message $m$, the receiver should not learn anything about $m$.
        \item\textit{Binding}: Given the commitment $\texttt{Com}(m)$ of a message $m$, the sender should only be able to open the commitment without changing the messages.
    \end{description}
\end{definition}

\subsubsection{Cryptographic protocol}\label{sec:smpc}

\begin{definition}[Secure Multi-party Computation]\cite{canetti2001universally}
Secure multi-party computation (MPC) allows $n$ distrust parties to jointly compute an agreed-upon function of their private inputs without revealing anything beyond the output.
\end{definition}

When executing a MPC protocol, parties communicate with each other (receiving messages and sending messages) as well as do some local computation, and they will obtain the output of the agreed-upon function after multiple rounds. \par

Informally, A MPC protocol is secure if participants cannot infer other party's input through inspecting the list of messages she sees during the execution\footnote{For simplicity, we omit cryptographic definition details here, and we will discuss them rigorously in appendix.}. 
The formalization of the above intuition is real world/ideal world paradigm \cite{goldreich1987play}. \par

\begin{definition}[Real World/Ideal World Paradigm] \label{def:realideal}
    The security of cryptographic protocol is formalized through the notion of comparison.\par
    \begin{description}
        \item\textit{Real world}: the parties execute a prescribed protocol
        \item\textit{Ideal world}: there is a trusted agent with complete privacy guarantee, and each party sends their inputs to this trusted agent and the trusted agent computes the agreed-upon function, and then sends the output back to each party
    \end{description}
\end{definition}

\begin{figure}[htbp]
    \centering
    \includegraphics[width=0.6\textwidth]{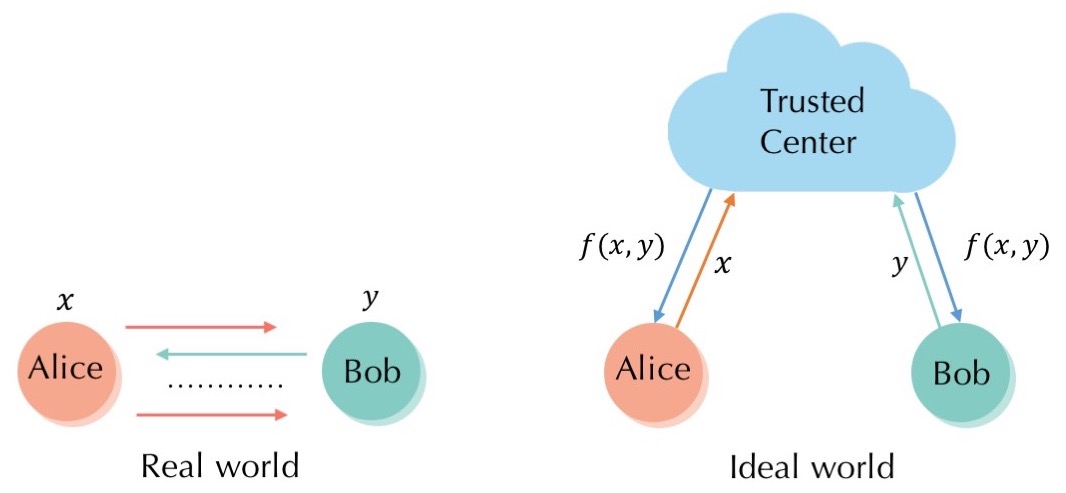}
    \caption{Real world/ideal world paradigm}
    \label{fig:my_label}
\end{figure}

With Definition~\ref{def:realideal}, a MPC protocol is secure if the scenario where parties executing the protocol is as secure as the scenario where parties has a trusted agent. The formalized notion of this is \textit{simulation}. We defer the detailed description of \textit{simulation} to appendix.\par

\begin{lemma}\label{lem:mpc}\cite{goldreich1987play}
There exists a secure multi-party computation protocol.
\end{lemma}

\subsection{Peer prediction}\label{sec:pp}
In the setting where information cannot be verified (e.g. subjective opinions, labels of medical images without known pathological truth), peer prediction focuses on designing reward schemes that help incentivize truthful, high quality information. In the peer prediction style information elicitation mechanisms, each information provider's reward only depends on her report and her peers' reports. 

Our SMind uses two special peer prediction mechanisms as building blocks that are proposed by \citet{2016arXiv160303151S,2016arXiv160501021K,kong2018water}. We believe that using other peer prediction mechanisms as building blocks is similar to our case. The information theoretic idea proposed by \citet{2016arXiv160501021K} can give a simple interpretation to the two special peer prediction mechanisms. At a high level, each information provider is paid the \emph{unbiased estimator} of the ``mutual information'', call it \emph{mutual information gain (MIG)}, between her reported information and her peer's reported information where the ``mutual information'' is \emph{information-monotone}---any ``data processing'' on the two random variables will decrease the ``mutual information'' between them. Thus, to get highest payment, the information provider should tell the truth, since the truthful information contains the most amount of information when applying non-truthful strategy can be seen as a ``data processing''. In our running example, a peer prediction style mechanism pays each hospital the ``mutual information'' between her labels and her peer's labels. 

We start to introduce the two special multi-task setting (the number of tasks must be $\geq 2$) peer prediction mechanisms. One is \emph{Correlation PP}, which elicits discrete signals (e.g. hard label: benign/malignant) and the other is \emph{Pearson PP}, which elicits forecasts (e.g. soft label: 70\% benign, 30\% malignant).

\paragraph{Correlation-PP/Pearson-PP}
Alice and Bob are assigned $N\geq 2$ priori similar questions (e.g. medical image labeling questions). \begin{description} 
\item[Report] for each question $i$, Alice is asked to report $r_A^{i}$ and Bob is asked to report $r_B^{i}$. We denote Alice and Bob's honest report for all questions by $\mathbf{r}_A$ and $\mathbf{r}_B$ respectively and denote their actual reports by $\reallywidehat{\mathbf{r}}_A$ and $\reallywidehat{\mathbf{r}}_B$. 

\item[Payment] Alice and Bob are rewarded by the average ``the amount of agreement'' between their reports in \emph{same} task, and punished the average ``the amount of agreement'' between their reports in \emph{distinct} tasks. When their reports are discrete signals, they are paid  
\begin{align} \label{eq:mig}
    MIG^{Corr}(\reallywidehat{\mathbf{r}}_A,\reallywidehat{\mathbf{r}}_B)=&\frac{1}{N}\sum_{i} \bigg(\mathbbm{1}(\reallywidehat{r}_A^i==\reallywidehat{r}_B^i)\bigg)\\ \nonumber
    &-\frac{1}{N(N-1)}\sum_{i\neq j}\bigg(\mathbbm{1}(\reallywidehat{r}_A^i==\reallywidehat{r}_B^j)\bigg)
\end{align}
When their reports are forecasts, given the prior distribution $\Pr[Y]$ (e.g. the apriori prediction of benign/malignant, like 90\% benign/10\% malignant), they are paid
\begin{align} \label{eq:mig}
    MIG^{Pearson}(\reallywidehat{\mathbf{r}}_A,\reallywidehat{\mathbf{r}}_B)=&\frac{1}{N}\sum_{i} 2\bigg(\sum_{y\in\Sigma}\frac{\reallywidehat{r}_A^i(y) \reallywidehat{r}_B^i(y)}{\Pr[Y=y]}-1\bigg)\\ \nonumber
    &-\frac{1}{N(N-1)}\sum_{i\neq j}\Bigg(\bigg(\sum_{y\in\Sigma}\frac{\reallywidehat{r}_A^i(y) \reallywidehat{r}_B^j(y)}{\Pr[Y=y]}\bigg)^2-1\Bigg) 
\end{align}
\end{description}

\begin{assumption}[A priori similar and random order]
    All questions have the same prior. All questions appear in a random order, independently drawn for both Alice and Bob.
\end{assumption}

Alice and Bob play a truthful strategy profile if they report $(\reallywidehat{\mathbf{r}}_A,\reallywidehat{\mathbf{r}}_B)=(\mathbf{r}_A,\mathbf{r}_B)$. Informally, Alice and Bob play a \emph{permutation strategy profile} if they both report the same permutation of their honest reports. (e.g. hard label case: label benign when it's malignant and say malignant when it's benign; soft label case: their honest forecasts are 70\% benign, 30 \% malignant and 65\% benign, 35\% malignant but they instead report  30\% benign, 70 \% malignant and 35\% benign, 65\% malignant). Here we omit the formal definition of permutation strategy profile as well as the introductions of other additional assumptions required by Correlation-PP and Pearson-PP since they are not the focus of this paper. Interested readers are referred to \citet{2016arXiv160303151S,2016arXiv160501021K,kong2018water}. We present the main property of Correlation-PP/Pearson-PP, i.e., both $ MIG^{Corr}(\reallywidehat{\mathbf{r}}_A,\reallywidehat{\mathbf{r}}_B)$ and $ MIG^{Pearson}(\reallywidehat{\mathbf{r}}_A,\reallywidehat{\mathbf{r}}_B)$ are maximized if $(\reallywidehat{\mathbf{r}}_A,\reallywidehat{\mathbf{r}}_B)=(\mathbf{r}_A,\mathbf{r}_B)$ and the maximum is always strictly positive, such that the sellers are willing to participate the game at the beginning. 

\begin{lemma}[Correlation-PP/Pearson-PP is truthful]\cite{2016arXiv160303151S,2016arXiv160501021K,kong2018water}\label{lem:corr}
With a priori similar and random order assumption, and mild conditions on the prior, Correlation-PP/Pearson-PP has truth-telling is a strict equilibrium and each agent's expected payment is strictly maximized when agents tell the truth, where the maximum is also strictly positive. Moreover, when agents play a non-permutation strategy
profile, each agent's expected payment is strictly less than truth-telling.
\end{lemma}

\subsection{Smart contract}\label{sec:smartcontract}
Smart contract enforces the execution of a contract between untrusted parties. It allows credible and irreversible transactions without a trusted third party. The assurance is based on the consensus protocol of blockchain.

In the most prominent smart-contract platform Ethereum, contract codes resides on the blockchain, executed on a decentralized virtual machine Ethereum Virtual Machine (EVM). Each instruction in a smart contract is ideally executed by all miners on the chain. Transaction fees provide economic incentives for miners to execute the contract, who pack transactions into blocks and record them on chain. 

It has been shown that without further assumption it is impossible to design protocols that guarantees complete fairness in exchange procedure without a trusted third party \cite{goldreich1987play, pagnia1999impossibility, yao1986generate}. The reliability of blockchain is ensured by honest majority on chain. In our protocol, smart contract provides support for the following functionalities:

\begin{description}
    \item[Ledger] A contract with $id$ stores ledger-needed information. It runs with multiple parties, stores their balance and frozen funds in the contract. 
    \item[Freeze Funds] It should be able to freeze funds from accounts onto chain. 
    \item[Unfreeze Funds] It should be able to unfreeze funds from the previously frozen to accounts in the contract. 
\end{description}

For simplicity, we do not elaborate on the functionalities of smart contract in our protocol. Instead, it is taken as a public bulletin board with code running on it.

\section{SMind: A Trust-free, Truthful, and Secure Information Trade Protocol}\label{sec:smartmpcpp}

In this section, we introduce our Smart Info-Dealer (SMind), a trust-free, truthful, and secure protocol that elicits unverifiable information. We first focus on the information trade setting where there are one buyer vs two sellers and will extend the setting to multiple sellers later. 

We recall our example of medical image labeling here and recommend the readers use this running example as a background when they read this section: a medical company (the buyer) wants to buy labels (hard label: benign/malignant, soft label: 90\% benign) for multiple difficult medical images with \emph{unknown pathological truth} from two hospitals (the sellers). 

We start by formally introducing the information trade setting and the definition of information trade protocol (Section~\ref{sec:model}). Then we will present three protocol design goals: trust-free, truthful, and secure (see informal definitions in Section~\ref{sec:intro}).  Figure~\ref{fig:smp} shows an overview of SMind. Finally, we will present the pseudosode of SMind (Section~\ref{sec:smart}) and show that it is trust-free, truthful, and secure (Section~\ref{sec:proof}). 


\begin{figure}
  \includegraphics[width=\linewidth]{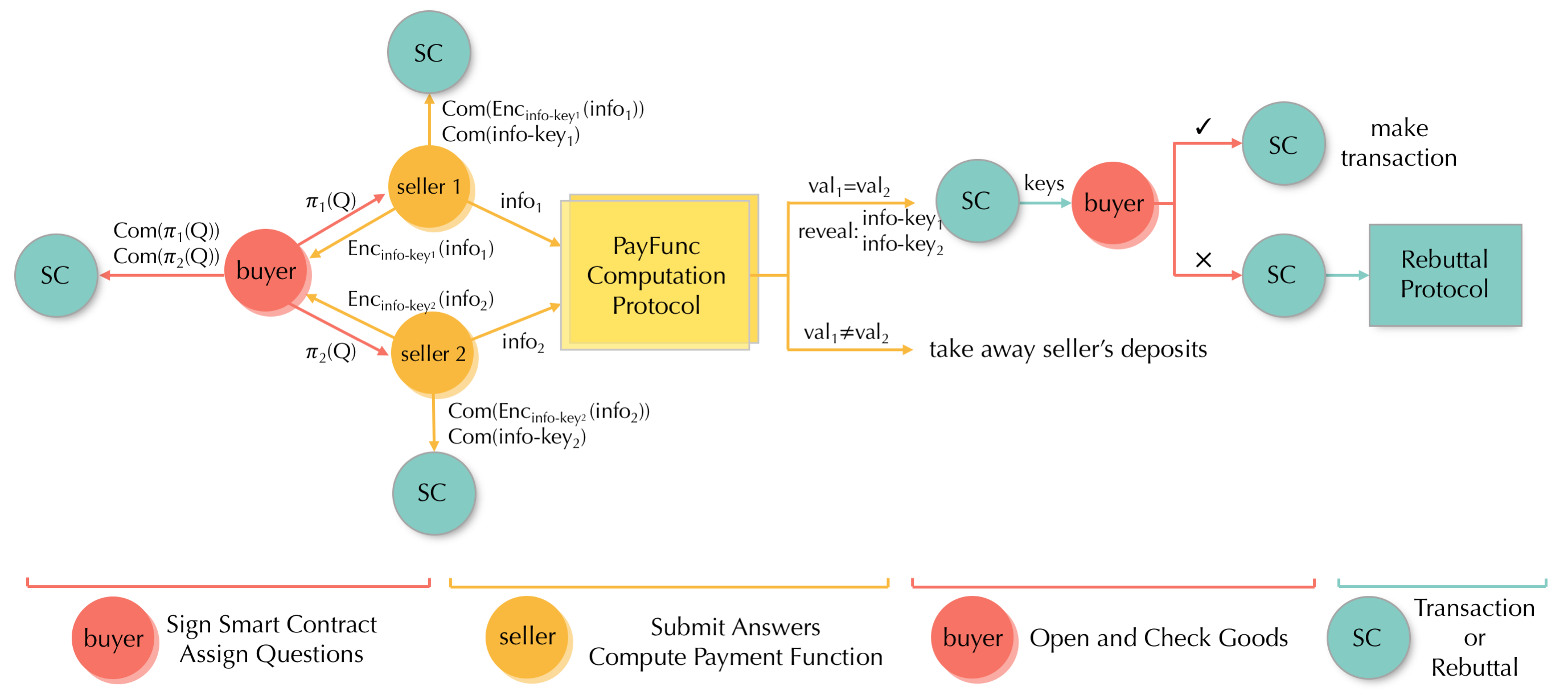}
  \caption{An illustration of SMind: A buyer (e.g. medical company) wants to buy two sellers' private opinions $\text{info}_1,\text{info}_2$ on several a priori similar tasks (e.g. soft or hard labels for several difficult medical images), with a fixed payment function $\textsc{PayFunc}(\text{info}_1,\text{info}_2)$ (e.g. a peer prediction style payment).  1) \textbf{Sign Smart Contract}: the buyers and sellers first sign the smart contract, submit the payment function, and pay their deposits (we omit this step in the figure). 2) \textbf{Assign Questions (red)}: the buyer assigns questions to the sellers privately in random order $\pi_1$ and $\pi_2$, and then commits to the question sets $\pi_i(Q)$ (\emph{each seller checks}) for possible future Rebuttal stage. 3) \textbf{Submit Answers (yellow)}: the sellers answer the questions individually. After they finish the questions, they package all information (include questions, answers) and send the encrypted info (by info-key) to the buyer and commit to info-keys and sends the commitment to smart contract. The sellers also commit to their encrypted info (\emph{buyer checks}) and send the commitment to the smart contract for possible future Rebuttal stage. 4) \textbf{Compute Payment Function (yellow)}: the sellers run MPC to compute the payment function and sends MPC's output individually to the public. 5) \textbf{Open and Check Goods (red)}: the sellers reveal their info-keys to the public. The buyer opens and checks the good. If the good is incorrect (wrong questions, or payment computation), then the buyer raises rebuttal. 6) \textbf{Transaction or Rebuttal (blue)}: if the buyer does not raise rebuttal, the smart contract makes transaction. Otherwise, the smart contract runs the rebuttal protocol with the information that is committed previously.}
  \label{fig:smp}
\end{figure}


\subsection{Model and setting}\label{sec:model}

\paragraph{Information trade setting $(\mathcal{B},\mathcal{S}_1,\mathcal{S}_2,N)$} A buyer $\mathcal{B}$ wants to buy two sellers' information (denoted as $\mathcal{S}_1,\mathcal{S}_2$), for $N$ ($N>1$) a priori similar events/questions $q_1,q_2,...,q_N$ (e.g. labeling medical images). The opinion format may be a discrete signal (e.g hard label: benign/malignant) in $\Sigma$ or a forecast (e.g. soft label: 70\% benign, 30\% malignant) in $\Delta_{\Sigma}$, where $\Delta_{\Sigma}$ is the set of all possible distributions over $\Sigma$, for the possible outcome of the event. We denote the honest private opinions of the two sellers by $(\mathbf{r}_1,\mathbf{r}_2)$ respectively and their actual reported opinions by $(\reallywidehat{\mathbf{r}}_1,\reallywidehat{\mathbf{r}}_2)$ respectively.

We call the buyer and sellers \emph{traders}. Each trader has a \emph{privacy cost}. Each seller's privacy cost $\textsc{PriCost}_{S_i}$ represents her cost when she knows her private information is revealed to other people besides the buyer who pays. The buyer's privacy cost $\textsc{PriCost}_{B}$ represents her cost when her bought information is revealed to other people (e.g.the public) besides the information owner. 

\begin{definition}[Information Trade Protocol (ITP)]\label{def:itp} Given a setting $(\mathcal{B},\mathcal{S}_1,\mathcal{S}_2,N)$, an information trade protocol is a protocol that allows the buyer $\mathcal{B}$ to buys two sellers' opinions $(\reallywidehat{\mathbf{r}}_1,\reallywidehat{\mathbf{r}}_2)$ with a fixed payment function $\textsc{PayFunc}(\reallywidehat{\mathbf{r}}_1,\reallywidehat{\mathbf{r}}_2)$. 
\end{definition}

\subsection{Protocol design goals: trust-free, truthful, and secure}\label{sec:goals}

We first give the formal definitions of trust-free and truthful ITPs. 

\begin{definition}[Trust-free ITP]
An ITP is trust-free if its execution does not need to assume that any trader is honest nor the existence of a trusted center. 
\end{definition}

In a trust-free ITP, the traders are allowed to be rational/strategic instead of required to be honest. To encourage the rational traders to behave honestly, a truth-free ITP should be additionally \emph{truthful}, which definition will be introduced now.  

Traders play \emph{truthful} strategy in ITP if they follow ITP honestly. At a high level, truthful sellers provide truthful information and truthful buyers pay the information with a fixed payment function. If there exists an equilibrium concept such that the ITP has truth-telling as the only equilibrium satisfying that equilibrium concept, we are convinced to say the traders will be encouraged to follow the ITP honestly, i.e. the ITP is \emph{truthful}. We pick \emph{strong SPE} (Definition~\ref{def:strongspe}) as the equilibrium concept. 

\begin{definition}[Truthful ITP]
An ITP is truthful if it has truth-telling as the only strong SPE.
\end{definition}
We give an informal definition of security here and will give a formal real world/ideal world style (see Definition~\ref{def:realideal}) definition in appendix. 
\begin{definition}[Secure ITP (informal)]
An ITP is secure if the information is only revealed to its owner and its buyer when traders follow the protocol. Except the output of the payment function, it's computationally infeasible for other people to obtain additional information, except a negligible probability.
\end{definition}

\subsection{SMind: description, assumptions, properties}\label{sec:smart}

We give the pseudocode of our SMind here (Table~\ref{protocol:mainbody},~\ref{protocol:checkrebuttal}). In the pseudocode, we use $\texttt{Op}(m)$ to denote the committed message $m$'s opening, but this does not mean $\texttt{Op}(m)$ depends on $m$, instead, it is chosen before generating the commitment of $m$.

\begin{table}[!htbp]
	\caption{SMind (Part I)}
	\label{protocol:mainbody}
	\begin{minipage}{\columnwidth}
		\begin{center}
			\begin{tabular*}{\textwidth}{l @{\extracolsep{\fill}}}
				\toprule
				
				\textbf{Sign Smart Contract}\\
	            Buyer and two sellers, $\mathcal{B},\mathcal{S}_1,\mathcal{S}_2$:\\
	            \qquad 1. determines and submits $\textsc{PayFunc}(\hat{\mathbf{r}}_1,\hat{\mathbf{r}}_2):=\alpha MIG(\hat{\mathbf{r}}_1,\hat{\mathbf{r}}_2)+\beta$ \footnote{The traders can pick $MIG$ as $MIG^{Pearson}$ or $MIG^{Corr}$ based on the format of the opinions (Section~\ref{sec:pp}).}\\
	            \qquad \qquad a. selects the coefficients $\alpha>0,\beta>0$ of the payment function\\ 
	            \qquad \qquad b. submits the coefficients to smart contract \\
	            \qquad \qquad c. submits the MIG function to smart contract\\
	            \qquad 2. pays their deposits $\textsc{Dep}_B, \textsc{Dep}_{S_1},\textsc{Dep}_{S_2}$ respectively to the contract\\

				\hline
				\textbf{Assign Questions}\\
				Buyer $\mathcal{B}$:\\
				\qquad 1. \textit{shuffles}\\
				\qquad\qquad a. randomly selects two permutations $\pi_1,\pi_2$ (a.k.a question order)\\
				\qquad\qquad b. generates two permuted question sets to get $\pi_1(Q), \pi_2(Q)$\\
				\qquad 2. \textit{commits}\\
				\qquad\qquad a. commits to $\pi_i(Q)$: generates randomness $\texttt{Op}(\pi_i(Q))$,\\
				\qquad \qquad generates the commitments $\texttt{Com}(\pi_i(Q)) \leftarrow Commit(\texttt{Op}(\pi_i(Q)), \pi_i(Q))$\\
				\qquad\qquad b. submits the two commitments $\texttt{Com}(\pi_1(Q)), \texttt{Com}(\pi_2(Q))$ to smart contract\\
				
				\qquad 2. \textit{waiting to be checked}\\
				\qquad \qquad sends $\pi_i(Q), \texttt{Op}(\pi_i(Q))$ to seller $i$.\\

	            Each Seller $\mathcal{S}_i$:\\
				\qquad \textit{checks the buyer's commitments}\\
				\qquad\qquad $checkbit \leftarrow Open(\texttt{Com}(\pi_i(Q)), \texttt{Op}(\pi_i(Q)), \pi_i(Q))$\\
				\qquad \qquad \texttt{if} $checkbit$ is $valid$, then goto \textbf{Submit Answers}\\
				\qquad \qquad \texttt{else} \textbf{terminates}\\

				\hline
				\textbf{Submit Answers}\\
				Each seller $\mathcal{S}_i$:\\
				\qquad 1. \textit{packages and encrypts}\\
				\qquad \qquad a. packages her ($\pi_i(Q)$, answers) as her $\text{info}_i$\\
				\qquad \qquad b. generates $\text{info-key}_i$, encrypts $\text{info}_i$ with info-key to get $\enc_{\text{info-key}_1}(\text{info}_i)$\\
				\qquad 2. \textit{commits}\\
				\qquad \qquad a. commits to $\enc_{\text{info-key}_i}(\text{info}_i)$:\\
                \qquad \qquad \qquad generates randomness $\texttt{Op}(\enc_{\text{info-key}_i}(\text{info}_i))$\\
				\qquad \qquad \qquad generates the commitment \\
				\qquad \qquad \qquad
				$\texttt{Com}(\enc_{\text{info-key}_i}(\text{info}_i)) \leftarrow
				Commit (\texttt{Op}(\enc_{\text{info-key}_i}(\text{info}_i)), \enc_{\text{info-key}_i}(\text{info}_i))$\\
				
				\qquad \qquad b. commits to $\text{info-key}_i$:\\
				\qquad \qquad \qquad generates randomness $\texttt{Op}(\text{info-key}_i)$\\
				\qquad \qquad \qquad generates the commitment $\texttt{Com}(\text{info-key}_i) \leftarrow Commit(\texttt{Op}(\text{info-key}_i),\text{info-key}_i)$\\
				\qquad \qquad c. submits $\texttt{Com}(\enc_{\text{info-key}_i}(\text{info}_i))$ and  $\texttt{Com}(\text{info-key}_i)$ to smart contract\\

				\qquad 3. \textit{waiting to be checked}\\
				\qquad \qquad sends $\enc_{\text{info-key}_i}(\text{info}_i)$ and 
				$\texttt{Op}(\enc_{\text{info-key}_i}(\text{info}_i))$ to buyer\\
				
				Buyer $\mathcal{B}$:\\
				\qquad \textit{checks the sellers' commitments}\\
				\qquad \qquad  $checkbit \leftarrow Open(\texttt{Com}(\enc_{\text{info-key}_i}(\text{info}_i)), \texttt{Op}(\enc_{\text{info-key}_i}(\text{info}_i)),\enc_{\text{info-key}_i}(\text{info}_i))$\\
				\qquad\qquad \texttt{if} $checkbit$ is $valid$, then goto \textbf{Compute Payment Function}\\
				\qquad\qquad \texttt{else} \textbf{terminates}\\
				%
				

				\bottomrule
			\end{tabular*}
		\end{center}
		\bigskip\centering
	\end{minipage}

\end{table}

\begin{table}[!htbp]
	\caption{SMind (Part II)}
	\label{protocol:checkrebuttal}
	\begin{minipage}{\columnwidth}
		\begin{center}
			\begin{tabular*}{\textwidth}{l @{\extracolsep{\fill}}}
				\toprule
			
				\textbf{Compute Payment Function}\\
				Two sellers $\mathcal{S}_1,\mathcal{S}_2$:\\
				\qquad 1. run MPC sub-protocol to compute the $\textsc{PayFunc}(\hat{\mathbf{r}}_1,\hat{\mathbf{r}}_2)$\\
				\qquad 2. submit output individually to smart contract\\
					\hline
				\textbf{Open and Check Goods}\\
				Each seller $\mathcal{S}_i$: \\
				\qquad submits $\text{info-key}_i$ and  $\texttt{Op}(\text{info-key}_i)$ to smart contract\\
				
				Buyer $\mathcal{B}$:\\
				\qquad \textit{checks}:\\
				\qquad \qquad use $\text{info-key}_i$ to decrypt $\enc_{\text{info-key}_i}(\text{info}_i)$\\
				\qquad \qquad \texttt{if} fail in decryption \\
				\qquad \qquad \qquad \texttt{or} questions in info$_i$ does not match $\pi_i(Q)$\\
				\qquad \qquad \qquad  \texttt{or} \textsc{PayFunc} is inconsistent as reported 
				\\
				\qquad \qquad \qquad raises \textbf{Rebuttal}\\

				Smart contract $\mathcal{SC}$:\\
				\qquad 1. \textit{waits for submission}\\
				\qquad \qquad \texttt{if} receives $\text{info-key}_i$s and consistent MPC outputs before time-out\\
				\qquad \qquad \qquad goto next step\\
				\qquad \qquad \texttt{else}\\
				\qquad \qquad \qquad confises all sellers' deposits, return buyer's deposits and \textbf{terminates}\\


				\qquad 2. \textit{checks}\\
				\qquad \qquad $checkbit\leftarrow Open(\texttt{Com}(\text{info-key}_i), \texttt{Op}(\text{info-key}_i), \text{info-key}_i)$\\
				\qquad \qquad \texttt{if} $checkbit$ is $invalid$\\
				\qquad \qquad \qquad  all deposits goes to buyer, \textbf{terminates}\\
			
				\qquad 3. \textit{makes transaction}\\
				\qquad\qquad \texttt{if} not receives \textbf{Rebuttal} (a.k.a time-out)\\
				\qquad\qquad\qquad make transactions and \textbf{terminates}\\

				\hline
				\textbf{Rebuttal}\\
				Buyer $\mathcal{B}$:\\
				\qquad submits $\enc_{\text{info-key}_i}(\text{info}_i)$,  $\texttt{Op}(\enc_{\text{info-key}_i}(\text{info}_i))$;
				 $\pi_i(Q)$ and $ \texttt{Op}(\pi_i(Q)), i=1, 2$\\
				
                Smart contract $\mathcal{SC}$:\\
                \qquad \texttt{if} not receives buyer's submissions (a.k.a. time-out)\\
                \qquad\qquad goto \textbf{rebuttal failure}\\
                
                \qquad $checkbit_1\leftarrow Open(\texttt{Com}(\enc_{\text{info-key}_i}(\text{info}_i)), \texttt{Op}(\enc_{\text{info-key}_i}(\text{info}_i)), \enc_{\text{info-key}_i}(\text{info}_i))$\\
                \qquad $ checkbit_2\leftarrow Open(\texttt{Com}(\pi_i(Q)), \texttt{Op}(\pi_i(Q)), \pi_i(Q))$\\
                \qquad \texttt{if} $checkbit_1$ or $checkbit_2$ is $invalid$, then goto \textbf{rebuttal failure}\\
                \qquad  \texttt{if} fail in decryption \\
				\qquad \qquad \texttt{or} questions in info$_i$ does not match $\pi_i(Q)$\\
				\qquad  \qquad  \texttt{or} \textsc{PayFunc} is inconsistent as reported \\
				\qquad \qquad goto \textbf{rebuttal success}
                
			

			    \\
			    \qquad \textbf{rebuttal failure}:
			     gives all deposits evenly to the sellers except the contract cost\\
			    \qquad \textbf{rebuttal success}: 
			     gives all deposits to the buyer except the contract cost \\
				\bottomrule
			\end{tabular*}
		\end{center}
		\bigskip\centering
	\end{minipage}

\end{table}

We present several reasonable assumptions of our main theorem.

\begin{assumption}\label{assume:id}
Initially, sellers do not know each other's identity.
\end{assumption}

This assumption guarantees that the sellers cannot privately communicate with each other \emph{before} the Compute Payment Function stage. We require this assumption to guarantee the truthful property of the peer prediction building block of SMind. Since without this assumption, 1) the sellers will play an \emph{order collusion} (e.g. answer yes/no for the questions with even/odd index) to get higher payments; 2) although permutation (e.g. label benign when it is malignant, label malignant when it is benign) cannot bring the sellers strictly higher payments (Lemma~\ref{lem:corr}), a permutation strategy profile is much less risky when the above assumption does not hold.   

The above assumption still allows the sellers to privately communicate with each other in the Compute Payment Function stage, since before this stage, the protocol has already asked the sellers to commit several necessary information securely for future possible Rebuttal stage.

\begin{assumption}\label{assume:notransfer}
Traders cannot transfer money after the protocol, aided by a trusted judge outside the protocol.
\end{assumption}

It may sound possible that the buyer can collude with one seller to cheat for all deposits and divide them evenly after the protocol. However, it implicitly requires a trusted judge to execute this, otherwise buyers will take all money and refuse to give her accomplice. 

To encourage the sellers to run MPC rather than calculate the payment in a non-private way (e.g. seller 1 sends her private information to seller 2 and seller 2 finishes all computations) in the Compute Payment Function stage, we need the following assumption. 
\begin{assumption}\label{assume:pri}
Both sellers have privacy costs that are greater than the cost of running MPC.
\end{assumption}

However, this assumption is not necessary if we do not care the computation method the sellers use, since all other parts (e.g. payment submission) in SMind are still truthful and secure without this assumption. 





\begin{theorem}[Main Theorem]
With assumption~\ref{assume:id},~\ref{assume:notransfer},~\ref{assume:pri}, there exists proper deposits such that SMind is trust-free, truthful, and secure. 
\end{theorem}

\section{Proof of Main Theorem} \label{sec:proof}

We recommend the readers to use Figure~\ref{fig:smp} as a reference when reading the proof. The trust-free property of SMind follows from its description. We show the truthful property and the secure property independently. 

To show SMind is truthful, i.e. has truth-telling as the unique strong SPE, we use backward induction procedure (Section~\ref{sec:gametheory}) and start from the last stage. We will firstly show that when deposits are large enough, the rational buyer will raise rebuttal as the protocol states. We then prove that the sellers' optimal strategy is to compute payment function for their committed answers and report $val_1=val_2=\textsc{PayFunc}(\reallywidehat{\mathbf{r}}_1,\reallywidehat{\mathbf{r}}_2)$ accordingly. Then we show that it is optimal for sellers to package and hash their honest answers $(\mathbf{r}_1,\mathbf{r}_2)$, due to the truthful property of the peer prediction style payment functions. Finally to the first stage when questions are assigned, we show that the rational buyer will follow the protocol honestly.

The security of SMind is based on security assumptions of its cryptographic building blocks including encryption, hash, commitment scheme and MPC. We only provide an intuition of proof here. For formal proof of security, readers can refer to Appendix \ref{sec:security}.

\subsection{Truthfulness proof: game theoretic analysis}

We start to show that with proper deposits, SMind has truth-telling as the unique strong SPE. We first list the possible costs in SMind: contract cost: \textsc{ConCost}/trader, additional rebuttal cost: \textsc{RebCost}, privacy cost: $\textsc{PriCost}_{trader}$, MPC cost: \textsc{MPCCost}/seller, (lower-bound) cost of attacking cryptographic building blocks: \textsc{AttackCost}. Note that the \textsc{AttackCost} is too large for a trader to ever try on a real attack. 


We identify SPE via a backward induction procedure: start from the last step, transactions or rebuttal. 

\subsubsection{Transactions or Rebuttal}

We first show that when the deposits are sufficiently large, the buyer's optimal strategy is to follow the protocol in the Open and Check Goods stage.

\begin{definition}[Incorrect good]\label{def:incorrect}
The information good is incorrect $\times$ if either of the following situations is true:
\begin{description}
\item[info-keys $\times$]: the revealed info-keys fail to open the encrypted info
\item[questions $\times$]: the questions set in the opened info is inconsistent with the committed questions set
\item[payment computation $\times$]: the payment of the opened info is inconsistent with the value two sellers submitted.
\end{description}
\end{definition}

\begin{lemma}[Optimal strategy in transaction or rebuttal: $\times \rightarrow \text{rebuttal}, \surd\rightarrow \text{no rebuttal}$]\label{lem:rebuttal}
There exists proper deposits, in detail, 
$$  \textsc{Dep}_{S_1}+\textsc{Dep}_{S_2}> \textsc{RebCost}+\textsc{PriCost}_B-val + 2 \textsc{ConCost}$$
and
$$  \textsc{Dep}_{B}>val + \textsc{ConCost}-\textsc{PriCost}_B \qquad \textsc{AttackCost} > \textsc{Dep}_{S_1}+\textsc{Dep}_{S_2}+\textsc{Dep}_{B} $$
such that after buyer opens and checks the information, it's optimal for the buyer to raise rebuttal when the good is incorrect and not raise rebuttal when the good is correct. 
\end{lemma}

\begin{proof}
We first claim that that the smart contract is able to open and check the goods with previously committed information, unless the buyer breaks the binding property of the commitment scheme. 
\begin{claim}\label{claim:rebuttal}
In the Rebuttal stage, 1) if the good is incorrect, the buyer will win; 2) if the good is correct, the buyer will lose unless she spends \textsc{AttackCost} to break the binding property of the commitment scheme. 
\end{claim}
Before the Rebuttal stage, the encrypted information, the key, the questions are all committed to the public (Figure~\ref{fig:smp}). If the good is incorrect, the buyer can share her view with the public by submitting the truthful encrypted information, such that the public can also know the good is incorrect. If the good is correct, the buyer will lose unless she breaks the binding property of the commitment scheme and submits a fake encrypted information. Thus, the above claim is valid. 

We start to show that the buyer's optimal strategy is $\times \rightarrow \text{rebuttal}, \surd\rightarrow \text{no rebuttal}$ via the following utility table, Table~\ref{tab:buyer}.

No rebuttal will always i) transfer $val$ from the buyer to the seller, ii) take the contract cost from the buyer.

If the good is wrong, rebuttal will i) return the buyer her own deposit except the contract cost, ii) bring the buyer all sellers' deposits except the contract costs, iii) take the buyer the rebuttal cost and her privacy cost. 

If the good is correct, if the buyer does not attack the commitment scheme, the rebuttal will i) take the buyer's deposit and ii) her privacy cost; otherwise the buyer will obtain the rebuttal benefits but lose the large attack cost.

\begin{table}[!htbp]%
	\caption{To rebuttal or not rebuttal?}
	\label{tab:buyer}
	\begin{minipage}{\columnwidth}
		\begin{center}
			\begin{tabular}{c|c|c}
				\toprule
				 & Rebuttal & no Rebuttal\\
								\hline
				$\times$ & $RB:=\textsc{Dep}_{S_1}+\textsc{Dep}_{S_2}-\textsc{RebCost}-\textsc{PriCost}_B-3\textsc{ConCost}$& $-val-\textsc{ConCost}$\\
				\hline
				 $\surd$ & $-\textsc{Dep}_B-\textsc{PriCost}_B$ or $\leq -\textsc{AttackCost}+\text{RB}$ & $-val-\textsc{ConCost}$ \\
				\bottomrule
			\end{tabular}
		\end{center}
		\bigskip\centering
	\end{minipage}
\end{table}%
The above table implies a proper deposits exist for the claim since the attack cost is very large. 
\end{proof}


%

\subsubsection{Compute Payment Function}

We move backward to Compute Payment Function stage and show that there exists proper deposits such that it is optimal for the sellers to honestly report $\textsc{PayFunc}(\hat{\mathbf{r}}_1,\hat{\mathbf{r}}_2)$, given that $(\hat{\mathbf{r}}_1,\hat{\mathbf{r}}_2)$ are the answers the sellers committed in the previous Submit Answers stage.

\begin{lemma}[Optimal strategy in Compute Payment Function stage: report $\textsc{PayFunc}(\hat{\mathbf{r}}_1,\hat{\mathbf{r}}_2)$]\label{lem:compute}
	Given that the buyer plays rationally in Transaction or Rebuttal stage, there exists proper deposits such that, it is optimal for both of the sellers to report $val_1=val_2=\textsc{PayFunc}(\hat{\mathbf{r}}_1,\hat{\mathbf{r}}_2)$ and to reveal keys honestly, given that $(\hat{\mathbf{r}}_1,\hat{\mathbf{r}}_2)$ are the answers the sellers committed in the previous Submit Answers stage.

	
	Moreover, if both the sellers' privacy costs are greater than the MPC cost, i.e., $$\textsc{PriCost}_{S_i}>\textsc{MPCCost},i=1,2$$ it's optimal for the sellers to run MPC to calculate the payment function. 
\end{lemma}

\begin{proof}
We first note that although the sellers can communicate with each other in this stage, they have no choice other than to compute the value of their committed data, report the value and reveal the keys honestly, otherwise they will lose either 1) large attack cost (much larger than the highest payment they can obtain in SMind) for breaking the commitment scheme, or 2) all their deposits in rebuttal, given that the buyer plays rationally in Transaction or Rebuttal stage. 


Moreover, if both the sellers' privacy costs are greater than the MPC cost, i.e., $\textsc{PriCost}_{S_i}>\textsc{MPCCost},i=1,2$, then it's optimal for the sellers to run MPC to calculate the payment function, otherwise although they may save the MPC cost, at least one of them will lose her privacy cost. 
\end{proof}

\subsubsection{Submit Answers}

We move back to the Submit Answers stage and will show that it's optimal for the sellers to package $\mathbf{r}_1,\mathbf{r}_2$ as their answers and encrypt and report the hash honestly. 

\begin{lemma}[Optimal strategy in Submit Answers stage: package $\mathbf{r}_1,\mathbf{r}_2$ as answers]\label{lem:finish}
In the Submit Answers stage, given that all traders will play rationally in the following stages, it is optimal for both of the sellers to 1) package $\mathbf{r}_1,\mathbf{r}_2$ as their answers and 2) encrypt all info (include questions set and answers) honestly; 3) commit the encrypted info honestly. It's also optimal for the buyer to check the sellers' commitments honestly. 
\end{lemma}

\begin{proof}
We start from the last step of this stage, the buyer checks the sellers' commitments. At this stage, the buyer cannot infer the private information (shown in the security proof). Thus, the rational buyer will check the commitment of encrypted information honestly i.e. the buyer will agree when it is correct (otherwise, the buyer loses the chance to buy the information) and disagree when it is wrong (otherwise, the buyer will lose in the Rebuttal stage).

We move backward to the sellers' parts: the information package and encrypted information commitment part. Given that all traders will behave rationally later (Lemma~\ref{lem:rebuttal},~\ref{lem:compute}), at this stage, the sellers must pick optimal $(\hat{\mathbf{r}}_1,\hat{\mathbf{r}}_2)$ to maximize $val=\textsc{PayFunc}(\hat{\mathbf{r}}_1,\hat{\mathbf{r}}_2)$. 

Lemma~\ref{lem:corr} shows that when the buyer assigns questions in a random order, $\textsc{PayFunc}(\hat{\mathbf{r}}_1,\hat{\mathbf{r}}_2)$ is maximized when $(\hat{\mathbf{r}}_1,\hat{\mathbf{r}}_2)=(\mathbf{r}_1,\mathbf{r}_2)$ and strictly maximized when there exists a permutation, such that $(\hat{\mathbf{r}}_1,\hat{\mathbf{r}}_2)$ is a permutation $(\mathbf{r}_1,\mathbf{r}_2)$. Note that Assumption~\ref{assume:id} guarantees that the sellers cannot communicate privately when they answer the questions and the sellers will prefer answer truthfully if the permutation strategy profile cannot bring them strictly more payments. Moreover, the hiding property of commitment scheme guarantees that the sellers cannot infer each other's order by public commitments of $\pi_1(Q)$ and $\pi_2(Q)$. 

Thus, it's optimal for the sellers to pick $(\hat{\mathbf{r}}_1,\hat{\mathbf{r}}_2)=(\mathbf{r}_1,\mathbf{r}_2)$ to maximize their payments $val$. Finally, we show that the sellers should package, encrypt and commit honestly. If the sellers do not package, encrypt and commit honestly, then it will either hurt the buyer or the sellers as SMind is almost a zero-sum game for the buyer and the sellers group. Rational seller will not hurt themselves and when it hurts the buyer, the rational buyers will disagree with the seller such that the sellers lose the chance to sell their information. Therefore, the rational sellers should package, encrypt and commit honestly. 
\end{proof}

\subsubsection{Assign Questions}

We move back to the initial stage, assign questions and will show that it is optimal for the buyer to follow the protocol honestly here. 

\begin{lemma}[Optimal strategy in Assign Questions stage: truthful strategy]\label{lem:assign}
In the Assign Questions stage, given that all traders will play rationally in the following stages, there exists proper deposits such that it is optimal for the buyer to follow the protocol honestly in this stage. 
\end{lemma}

\begin{proof}
We start from the last step of the Assign Questions stage, the sellers check the correctness of committed questions $\texttt{Com}(\pi_1(Q))$ and $\texttt{Com}(\pi_2(Q))$. If they ignore the inconsistency between their questions and the commitments, then they will lose their deposits in rebuttal. If they wrongly disagree with the correct commitments, the contract will be rescinded, then they will waste \textsc{ConCost} and lose the chance to sell their information. Therefore, with sufficiently large deposits, the rational sellers will check the committed questions honestly. 

When the buyer behave dishonestly in the Assign Question stage, it is possible that 1) the buyer does not commit properly, i.e. apply the incorrect commitment scheme or apply the correct commitment scheme but commit the wrong questions; 2) the buyer does not assign the questions set properly, for instance, not in a random order. 

Both of the cases will either hurt the buyer or the seller, since by thinking the sellers as an unit, SMind is almost a zero-sum game. Rational buyer will not hurt herself and if it hurts the seller, then the rational seller will disagree for this commitment. Thus it is optimal for the buyer to follow the protocol honestly in the Assign Questions stage. 
\end{proof}

After the above analysis, We are ready to finish our truthfulness proof. First we can see truth-telling is a SPE, but not a unique SPE since the sellers report the same but wrong value in the Compute Payment Function stage can also consist of a ``bad'' SPE. However, we note that these ``bad'' SPEs are not strong NE since the sellers can together deviate to the truthful strategy profile to benefit both of them. Thus, truth-telling is the unique strong SPE here.

\subsection{Security proof: cryptographic analysis}

In this section, we prove that honest-but-curious participants who follow the protocol cannot learn additional information from the other traders.\par

\begin{description}
    \item[Security against buyer]
    Security against buyer means that buyer should not learn additional information about seller's data before the sellers reveal the keys. Note that before revealing the keys, the buyer only has the encryption of the private information and the commitment of the keys. Then, based on the security of the encryption scheme and the hiding property of the commitment, SMind has security against buyer.

    \item[Security against seller]
    Security against seller means that a seller should not learn additional information about other seller's data (answers) during the whole protocol. From the security of MPC, a seller cannot learn any additional information of other seller's input by inspecting the communication transcript (messages being sent and received by a single seller during the MPC). From the hiding property of  commitment scheme, a seller cannot infer any additional information of the ciphertext of other seller's input (after she gets the decryption key, she cannot get other seller's input anyway).
    
    \item[Security against public] It means that the public should not learn information except the value of \textsc{PayFunc} submitted to the smart contract. After the keys are revealed, they cannot infer additional information anyway. They can only see the commitment of encrypted data. From the hiding property of commitment scheme, the public does not know the exact value of encrypted data. 
    
\end{description}


\section{Extension to multiple sellers}\label{sec:ext}

This section introduces a natural extension of SMind to the setting where there are multiple sellers. There are only two main differences. One is the payment function for each seller, and the other is the output of MPC protocol. 

We use $\reallywidehat{\mathbf{r}}_{-i}$ to denote the set of all sellers' answers excluding seller $i$. In the information trade setting with multiple sellers, the buyer pays each seller $i$ for 

$$\textsc{PayFunc}(\reallywidehat{\mathbf{r}}_i;\reallywidehat{\mathbf{r}}_{-i}):=\alpha \sum_{j,j\neq i} MIG(\reallywidehat{\mathbf{r}}_i,\reallywidehat{\mathbf{r}}_j)+\beta$$

Then, in the Compute Payment Function stage, all sellers run MPC protocol to output a payment vector $\mathbf{val}$ such that $\mathbf{val}(i)=\textsc{PayFunc}(\reallywidehat{\mathbf{r}}_i;\reallywidehat{\mathbf{r}}_{-i})$. If all sellers cannot reach an agreement on the payment vector, then their deposits are taken. Once they reach an agreement, they can reach the Transaction or Rebutal stage, like the two sellers' version. 

By going through all proofs carefully, this multiple sellers version SMind is still trust-free, truthful, and secure.

Compared with the two sellers version, the multiple sellers version SMind is more desirable in applications, due to the diversity in the payments. If there are three sellers, two high-quality, one low-quality, then the low-quality seller will be paid poorly since she has poor correlation with both other sellers, while the high-quality seller will be paid fairly since she has high correlation with another high-quality seller. Moreover, the low-quality seller cannot get other high-quality sellers' information, due to the security of SMind.

\section{Conclusion and Future Work}\label{sec:futurework}
In an unverifiable information trade scenario, we propose a trust-free, truthful, and secure information trade protocol, \emph{SMind}, by borrowing three cutting-edge tools that include peer prediction, secure multi-party computation, and smart contract.

A limitation of SMind is the lack of robustness. For simplicity of the game theoretic analysis, we let the sellers play a coordination game in the Compute Payment Function stage. However, if one of the sellers is irrational, this will lead to bad results for all sellers. One future direction is to design the protocol more delicately to make it robust. 

Another direct future work is to implement our protocol in smart contract to allow unverifiable information trade in different scenarios, for instance, the data trade in machine learning scenario.


\bibliographystyle{ACM-Reference-Format}
\bibliography{reference}

\newpage 
\appendix
\section{Cryptographic building blocks}\label{sec:cryptobb}
\subsection{Programmable Global Random Oracle}\label{sec:pgro}

In our proof, we use a restricted programmable and obeservable global random oracle, a model of a perfect hash function returning uniformly random values. All parties in the protocol can query the global random oracle. Following the works of \cite{camenisch2018wonderful, dziembowski2018fairswap}, we are able to use encryption scheme and commitment scheme with this random oracle in our protocol. Figure \ref{tab:randora} shows the ideal functionalities of the random oracle. 

\begin{table}[!htb]
	\caption{The ideal functionality of restricted programmable and observable random oracle, denoted as $\mathcal{H}$. }
	\label{tab:randora}
	\begin{minipage}{\columnwidth}
		\begin{center}
			\begin{tabular}{ll}
				\toprule
				\textbf{Parameter}:  output size $\mu$\\
				\textbf{Variable}: initially empty lists $L$, $P$ storing queries and programmed queries, respectively, \\
				and a list $L_{\mathcal{A}}$ storing illegitimate queries made by the adversary.\\
				1. \textbf{Query}:\\
				Upon receiving a query with $(\texttt{query}, q, \textsc{id})$ from party of session $\textsc{id}'$:\\
				\qquad if $(q, r, \textsc{id})\in L$, respond with $r$.\\
				\qquad if $(q, \textsc{id})\notin L$, sample $r\in \{0, 1\}^\mu$, store $(q, r, \textsc{id})$ in $L$, and respond with $r$.\\
				\qquad if $\textsc{id}\neq\textsc{id}'$, store $(q, r)$ in $L_{\mathcal{A}}$, respond with $r$\\
				2. \textbf{Program}:\\
				Upon receiving a program instruction $(\texttt{program}, q, r, \textsc{id})$ from the adversary $\mathcal{A}$:\\
				\qquad if $( q, r, \textsc{id})\in L$, abort. \\
				\qquad else, store $(q, r, \textsc{id})$ in $L$ and $(q, \textsc{id})$ in $P$\\
				Upon receiving $(\texttt{isProgrammed}, q)$ from a party of session $\textsc{id}$:\\
				\qquad if $(q, \textsc{id})\in P$, respond with 1\\
				\qquad else respond with 0\\
	            3. \textbf{Observe}:\\
	            Upon receiving $(\texttt{observe})$ from the adversary, respond with $L_\mathcal{A}$.
							\\
				\bottomrule
			\end{tabular}
		\end{center}
		\bigskip\centering
	\end{minipage}
\end{table}%

The programmability of global random oracle provides a strong power for the simulator in proof. \emph{It enables the simulator to send the buyer a garbage encryption or a garbage commitment without knowing the real encrypted or committed message and then program the random oracle to decrypt or open the previous garbage encryption or commitment to the real message afterwards when the real message is revealed} (see details in \cite{dziembowski2018fairswap}). We will employ this property of programmbility several times in the future security proof. Note that the simulator can succeed except with negligible probability.

\subsection{Encryption Scheme}
We firstly define the security of a encryption scheme (IND-CPA secure), and then give a secure encryption algorithm. 

Intuitively speaking, IND-CPA secure symmetric encryption guarantees that the encryption of two strings is indistinguishable, and further, an adversary (e.g. an eavesdropper in the communication channel between buyer and sellers) cannot distinguish the ciphertext of $m_0, m_1$ chosen by her in any case.

\begin{definition}indistinguishability under chosen ciphertext attack (IND-CPA)\par
Let $\mathcal{E} = (\text{Gen}, \enc, \dec)$ be a symmetric encryption scheme, and let $\mathcal{A}$ be an polynomial adversary who has access to an oracle. 
On input $(m_0, m_1) $, $LR(\cdot,\cdot, b)$ responds with $\enc_k(m_b)$.
We consider the following two experiments:\par
$\textbf{Expt}_{\mathcal{A}}^0 : k\leftarrow \text{Gen}(1^n), b\leftarrow \mathcal{A}^{LR(\cdot, \cdot, 0)}, \text{return } b$\par

$\textbf{Expt}_{\mathcal{A}}^1 : k\leftarrow \text{Gen}(1^n), b\leftarrow \mathcal{A}^{LR(\cdot, \cdot, 1)}, \text{return } b $\par
Then $\mathcal{E}$ is IND-CPA secure encryption scheme if 
$$ |Pr[\textbf{Expt}_{\mathcal{A}}^0 =1] - \Pr[\textbf{Expt}_{\mathcal{A}}^1 =1 ] |\le \mu(n)$$
\end{definition}

In the previous section \ref{sec:crypt} we introduce the definition of a encryption scheme. Here we provide the encryption algorithm that is IND-CPA secure in Algorithm \ref{alg:enc} and \ref{alg:dec} \cite{dziembowski2018fairswap}.

\begin{algorithm}[!htb]
\DontPrintSemicolon
	\SetAlgoNoLine
	\KwIn{message $m=(m_1, m_2, \dots, m_n)$, s.t. $\forall i, |m_i|=\mu$}
	$k\leftarrow \{0, 1\}^\kappa$, s.t. $\forall i, \mathcal{H}(\texttt{isProgrammed}, k\|i)\neq 1$\;
	\For{$i=1$ \KwTo $n$}{
	$k_i = \mathcal{H}(k\|i)$\;
	$\enc(m)_i=k_i\oplus m_i$}
	\KwOut{$\enc(m)=(\enc(m)_1, \dots, \enc(m)_n)$}
	\caption{Algorithm for encryption $\enc$.}
	\label{alg:enc}
\end{algorithm}

\begin{algorithm}[!htb]
\DontPrintSemicolon
    \SetAlgoNoLine
    \KwIn{ciphertext $c=(c_1, \dots, c_n)$, s.t. $|c_i|=\mu$, and key $k$}
    \For{$i\leftarrow1$ \KwTo $n$}{
    $k_i=\mathcal{H}(k\|i)$\;
    \eIf{$\mathcal{H}(\texttt{isProgrammed}, k\|i)$}{\KwOut{$\bot$}}
    {$m_i=k_i\oplus c_i$}}
    \KwOut{$m=(m_1, m_2, \dots, m_n)$}
    \caption{Algorithm for decryption.}
    \label{alg:dec}
\end{algorithm}

\subsection{Commitment Scheme}
Recall the definition of a commitment scheme in section \ref{sec:crypt}. Here we provide the commitment scheme used in our proof, contructed with the global random oracle.

Algorithm \ref{alg:commit} shows that if $\texttt{Op}(m)$ is long enough, the chance for $\mathcal{A}$ to distinguish $\texttt{Com}(m_1)$ and $\texttt{Com}(m_2)$ is negligible. Thus the hiding property is guaranteed. Algorithm \ref{alg:opencommit} shows how to open the commitment.
It is also hard to break the binding property, since $\mathcal{H}$ outputs uniformly distributed random numbers. 

\begin{algorithm}[!htb]
\DontPrintSemicolon
	\SetAlgoNoLine
	\KwIn{message $m\in \{0, 1\}^*$}
	$\texttt{Op}(m))\leftarrow \{0, 1\}^\kappa$, s.t. $\mathcal{H}(\texttt{isProgrammed}, m\|\texttt{Op}(m))\neq 1$\;
	$\texttt{Com}(m)\leftarrow \mathcal{H}(m\|\texttt{Op}(m))$\;
	\KwOut{$(\texttt{Com}(m), \texttt{Op}(m))$}
	\caption{Algorithm for commitment \texttt{Com}.}
	\label{alg:commit}
\end{algorithm}

\begin{algorithm}[!htb]
\DontPrintSemicolon
	\SetAlgoNoLine
	\KwIn{$\texttt{Com}(m)$, $\texttt{Op}(m)$}
	\eIf {$\texttt{Com}(m)=\mathcal{H}(m\|\texttt{Op}(m))$ \textbf{and}  $\mathcal{H}(\texttt{isProgrammed}, m\|\texttt{Op}(m))=0$}{
	\KwOut{1}}
	{
	\KwOut{0}
	}
	\caption{Algorithm for opening commitment.}
	\label{alg:opencommit}
\end{algorithm}

\subsection{Security definition: simulation for semi-honest adversaries}

Simulation is a way of comparing what happens in the ``real world'' to what happens in an ``ideal world''. The ``ideal world'' is usually secure by definition. For any adversary in real world, if there exists a simulator in ideal world who can achieve almost the same attack as adversary in real world, then the protocol is said to be secure.

We start to formally define security in the presence of semi-honest, ``honest-but-curious'', adversaries, i.e. the definition of $t$-private \cite{asharov2011full}. At a high level, a protocol is $t$-private if the view of up to $t$ corrupted parties in a real protocol execution can be generated by a simulator.

For a protocol $\Pi$, the view of the each party $i$, $\view^{\Pi}_i(\mathbf{x})$ is defined as her inputs, internal coin tosses and the messages she receives during an execution of a protocol $\Pi(\mathbf{x})$. For each subset of parties $I$, $\view^{\Pi}_{I}$ is defined as the union views of all parties in $I$. 

\begin{definition}[$t$-privacy of $m$-party protocols \cite{asharov2011full}]
Let $f$ be a a probabilistic $m$-ary functionality that maps $m$ inputs to $m$ outputs, i.e., $(\{0,1\}^*)^m \mapsto (\{0,1\}^*)^m$ and let $\Pi$ be a protocol. $\Pi$ is \emph{$t$-private} for $f$ if for every $I\subset [m]$ of cardinality at most $t$, $\Pi$ is \emph{$I$-private} for $f$ in the sense that there exists a probabilistic polynomial-time algorithm $\simu$ such that for every input $\mathbf{x}$ it holds that:
\begin{description}
\item \textit{Privacy}: $\simu(\text{security parameter},I\text{'s inputs and outputs})$ is computationally indistinguishable with $I$'s views in protocol $\Pi$ given input $\mathbf{x}$, i.e., $\view_I^{\Pi}(\mathbf{x})$;
\item \textit{Correctness}: when $f$ is a deterministic function, $f(\mathbf{x})$ equals $\Pi$'s output $\text{output}^{\Pi}(\mathbf{x})$, when $f$ is a probabilistic function, $f(\mathbf{x})$'s distribution equals $\text{output}^{\Pi}(\mathbf{x})$'s distribution. 
\end{description}
$\Pi$ is full-private for $f$ when $\Pi$ is $m$-private for $f$. 
\end{definition}

We introduce a common technique, sequential composition, in security proof that will simplify the proof substantially. Informally, it says that if a protocol $\Pi$'s sub-protocols are secure and satisfy a ``sequential and isolated'' condition, the protocol $\Pi$ is secure as well. 

\begin{lemma}[Sequential Composition\cite{canetti2001universally}]\label{lem:sq} 
    Let $g_i$ denote ideal functionality, and let $\Pi^{\rho_1,...}$\footnote{$\Pi^{\rho_1,...,\rho_n}$ denotes a protocol $\Pi$ which calls ideal functionality $\rho_1,...\rho_n$ during its execution.} denote the real world protocol such that 1) sub-protocol $\rho_1,...$ are called sequentially, 2) no $\Pi$-messages sent while $\rho_1, ... $ are executed. If $\Pi^{g_1,...}$ is a secure hybrid world protocol computing $f$ and each $\rho_i$ securely computes $g_i$, then $\Pi^{\rho_1,...} $ is a secure real-world protocol computing $f$.  
\end{lemma}

\section{Formal security proof}\label{sec:security}

We present our formal security proof here. Note that SMind motivates all traders to behave honestly. Thus, we only define and prove the security against semi-honest, i.e., ``honest-but-curious'', adversaries. We will first give the definition of an ideal ITP which is secure by definition and then show that there exists simulators that can simulate SMind only with views from the ideal ITP, which implies that SMind is secure, as SMind reveals no more information than the ideal ITP. 

For simplicity of the proof, we employ the sequential composition technique here and first consider a middle protocol $\Pi$. We define protocol $\Pi$ as the protocol that is the same as SMind except the Compute Payment Function stage. In $\Pi$'s Compute Payment Function stage, the sellers submit their infos to a trusted center $\mathcal{M}$ and $\mathcal{M}$ finishes the payment computation privately and reveals the output publicly. Thus, $\Pi$ replaces the MPC session in SMind by a trusted center. We will show that $\Pi$ is computationally indistinguishable with the ideal world. Then according to Lemma~\ref{lem:sq} and Lemma~\ref{lem:mpc}, SMind is also computationally indistinguishable with the ideal world.

\begin{definition}[Ideal ITP]
An ideal ITP's functionality:  
\begin{description}
\item[Sign Contract] Upon receiving deposits and the consistent payment functions from the traders, reveals the contract publicly. 
\item[Assign Questions] Upon receiving question set $Q$ from buyer, generates two random orders $\pi_1,\pi_2$ and sends $\pi_i(Q)$ to each seller $\mathcal{S}_i$ privately.
\item[Submit Answers] Waits for the sellers to finish the questions and collects each seller's info privately. 
\item[Compute Payment Function] Upon receiving info$_i$ from two sellers, computes \textsc{PayFunc}, and reveals the output to public.
\item[Make Transactions] Sends infos to the buyer and makes transactions based on the contract. 
\end{description}
\end{definition}

\begin{definition}[Security ITP (formal)]
An ITP is secure if the ITP fulfills the ideal ITP's functionality, and \emph{in every stage}, ITP is full-private for ideal ITP. 
\end{definition}

We start to show the middle protocol $\Pi$ is secure by enumerating all possible subsets $I$ and show that $\Pi$ is $I$-private for ideal ITP in every stage. 

\paragraph{Notations} We use the subscript of protocol $\Pi$ to denote the stages $\Pi$ is in. For instance, $\Pi_{s-c}$ represents the protocol execution from Sign Contract stage to Compute Payment Function stage. When the real message is $m$, we sometimes use $\reallywidehat{m}$ to denote the simulated $m$. For instance, we use $\reallywidehat{\texttt{Com}(m)}$ to denote the simulation of $\texttt{Com}(m)$. Although $\reallywidehat{\texttt{Com}(m)}$ has $m$, $\reallywidehat{\texttt{Com}(m)}$ is usually a random string that is generated \emph{independent of} $m$, i.e., without the knowledge of $m$. 

In the simulation, we assume the simulator in the ideal world can simulate the smart contract internally. The simulation for smart contract is very similar with \cite{dziembowski2018fairswap} and we omit the detailed description of the simulation in our proof.

\subsection{Security against the smart contract/public: $\mathcal{SC}$-private}
Before the Compute Payment Function stage, the public's view in $\Pi$ is 
$$ \view_{\mathcal{SC}}^{\Pi_{s-c}}=(\{\texttt{Com}(\pi_i(Q)),\texttt{Com}(\enc_{\text{info-key}_i}(\text{info}_i)),\texttt{Com}(\text{info-key}_i)|i=1,2\}) $$
However, in ideal world, a smart contract neither sends nor receives message from the ideal functionality. 

At a high level, the simulator can replace $\text{info}_i$ by garbages, generate keys, openings and then simulate the above view accordingly. In detail, The simulator can simulate the view by generating random $\reallywidehat{\text{info-key}_i},\reallywidehat{\texttt{Op}(\text{info-key}_i)}, i=1,2$ and then outputting

$$ \{(\reallywidehat{\texttt{Com}(\pi_i(Q))},\texttt{Com}(\enc_{\reallywidehat{\text{info-key}_i}}(1^{|\text{info}_i|})),\texttt{Com}(\reallywidehat{\text{info-key}_i})|i=1,2\}) $$
to simulate $\view_{\mathcal{SC}}^{\Pi_{s-c}}$ respectively, where $\reallywidehat{\texttt{Com}(\pi_i(Q))}$ denotes the simulated commitment for question set. But in fact, $\reallywidehat{\texttt{Com}(\pi_i(Q))}$ is a random string that is generated \emph{without} the knowledge of $\pi_i(Q)$. Since both our commitment scheme is constructed based on the global random oracle, the output of the simulator is computationally indistinguishable with the real view, only except a negligible probability, according to the property of the global random oracle.

In the Compute Payment Function stage, the public additionally views the output of the payment function while the simulator is also given this output in the ideal world. 

After the Open and Check Goods stage, the public in $\Pi$ additionally has $$\{(\texttt{Op}(\text{info-key}_i),\text{info-key}_i)|i=1,2\}.$$ 

The simulator can reveal the previously used 
$$\{(\reallywidehat{\texttt{Op}(\text{info-key}_i)},\reallywidehat{\text{info-key}_i})|i=1,2\}$$ without being distinguished since it is consistent with the simulator's previous output.

\subsection{Security against the buyer: $\mathcal{B}$-private}

Before the Open and Check Goods stage, the buyer's view in $\Pi$ is 
$$\view_{\mathcal{B}}^{\Pi_{s-c}}=(Q,\{\pi_i,\texttt{Op}(\pi_i(Q)),\enc_{\text{info-key}_i}(\text{info}_i),\texttt{Op}(\enc_{\text{info-key}_i}(\text{info}_i)),\texttt{Com}(\text{info-key}_i)|i=1,2\}).$$

We start to construct the simulator $\simu_{\mathcal{B}}$ to simulate the buyer's view with only the input message $Q$ and security parameter. 

$\simu_{\mathcal{B}}$ will first generates two random orders $\reallywidehat{\pi_1},\reallywidehat{\pi_2}$ and four randomness $$\reallywidehat{\texttt {Op}(\pi_i(Q))},\reallywidehat{\texttt{Op}(\enc_{\text{info-key}_i}(\text{info}_i))},i=1,2$$ as openings. Note here we do not need to know the message to generate its opening since we just use $\texttt{Op}(m)$ to denote the opening for $m$'s commitment and in fact, the randomness $\texttt{Op}(m)$ is independent of $m$. 

Then $\simu_{\mathcal{B}}$ generates four random strings $\reallywidehat{\texttt{Com}(\text{info-key}_i)}, i=1,2$, $\reallywidehat{\enc_{\text{info-key}_1}(\text{info}_i)}, i=1,2$. The four random strings are generated independently and uniform at random \emph{without any knowledge of infos or info-keys}. \emph{We will use the programmability of our global random oracle to program these four random strings such that they will correspond to future revealed real info-keys and infos} (Section~\ref{sec:pgro}). Since both our commitment scheme and encryption scheme are constructed based on the global random oracle, the output of $\simu_{\mathcal{B}}$ is computationally indistinguishable with $\view_{\mathcal{B}}^{\Pi}$, only except a negligible probability, according to the property of the global random oracle.

After the Open and Check Goods stage, the buyer's view in $\Pi$ additionally has $$\view_{\mathcal{B}}^{\Pi_o}=\{(\texttt{Op}(\text{info-key}_i),\text{info-key}_i)|i=1,2\}.$$ 
In the ideal world, the simulator $\simu_{\mathcal{B}}$ will know infos in this stage. To simulate the view in $\Pi$, we construct info-keys $\reallywidehat{\text{info-key}_i},i=1,2$ and employ the the programmability of the global random oracle (Section~\ref{sec:pgro}) such that the previously generated garbage $\reallywidehat{\enc_{\text{info-key}_i}(\text{info}_i)}, i=1,2$ can be decrypted to real infos and the previously committed garbage $\reallywidehat{\texttt{Com}(\text{info-key}_i)}, i=1,2$ can be opened as $\reallywidehat{\text{info-key}_i},i=1,2$, with openings $\reallywidehat{\texttt{Op}(\text{info-key}_i)}$. Then the simulator outputs $(\reallywidehat{\texttt{Op}(\text{info-key}_i)},\reallywidehat{\text{info-key}_i})$. Based on the construction of the outputs, these outputs are indistinguishable with the buyer's real view in $\Pi$ since both of them are consistent with the previous views in their worlds respectively.

\subsection{Security against one seller: $\mathcal{S}_i$-private}
Before the Compute Payment Function stage, each seller $i$'s view is 
\begin{align*}
\view_{\mathcal{S}_i}^{\Pi_{s-s}}=&(\pi_i(Q),\texttt{Op}(\pi_i(Q)), \texttt{Com}(\pi_1(Q)),\texttt{Com}(\pi_2(Q)),\text{info}_i,\\ &\texttt{Op}(\enc_{\text{info-key}_i}(\text{info}_i)), \text{info-key}_i, \texttt{Op}(\text{info-key}_i),\\
&\texttt{Com}(\text{info-key}_{-i}),\texttt{Com}(\enc_{\text{info-key}_i}(\text{info}_{-i})))
\end{align*}

In this stage, $\simu_{\mathcal{S}_i}$ is given $\pi_i(Q),\text{info}_i$. The simulator will simulate the view by using the given inputs and additionally generating several independent random strings to simulate other information (we still add wide hat to the real views to denote the simulated ones), which is computationally indistinguishable with the real world view only except a negligible probability, according to the property of the global random oracle.

In the Compute Payment Function stage, the seller will additionally view the output of the payment function, which is also viewed by the simulator in the ideal world. Thus, the seller's view can still be simulated here. 

After the Open and Check Goods stage, the seller $i$ needs to submit $\text{info-key}_i$ and $\texttt{Op}(\text{info-key}_i)$ and additionally view
$$ (\texttt{Op}(\text{info-key}_{-i}),\text{info-key}_{-i})$$ 
In the ideal world, the simulator $\simu_{\mathcal{S}_i}$ does not gain additional knowledge as the trusted center only sends the infos to the buyer. But the simulator can still simulate the additional view by employing the the programmability of the global random oracle and output $(\reallywidehat{\texttt{Op}(\text{info-key}_{-i})},\reallywidehat{\text{info-key}_{-i}})$ that are consistent with previously outputted $\reallywidehat{\texttt{Com}(\text{info-key}_{-i})}$.

\subsection{Security against two sellers: $\{\mathcal{S}_1,\mathcal{S}_2\}$-private}
Before the Compute Payment Function stage, two sellers' view is 
\begin{align*}
(\view_{\mathcal{S}_1}^{\Pi_{s-s}},\view_{\mathcal{S}_2}^{\Pi_{s-s}})
\end{align*}

Similar with the one seller setting, the simulator $\simu_{\mathcal{S}_1,\mathcal{S}_2}$ is given $\pi_i(Q),\text{info}_i,i=1,2$ and can simulate other information in real views by generate several independent random strings. The Compute Payment Function stage's simulation is also the same as the one seller case. From now, since all traders are semi-honest, the protocol will finish in the Open and Check Goods stage and the two sellers will not observe new information. Thus, the simulation is finished.

\subsection{Security against the buyer and one seller: $\{\mathcal{B},\mathcal{S}_i\}$-private}
Before the Compute Payment Function stage, the buyer and one sellers' view in $\Pi$ is 
\begin{align*}
\view_{\mathcal{B},\mathcal{S}_i}^{\Pi_{s-s}}=&(Q,\pi_1,\pi_2, \texttt{Op}(\pi_1(Q)),\texttt{Op}(\pi_2(Q)),\text{info}_i, \texttt{Op}(\enc_{\text{info-key}_i}(\text{info}_i)),\text{info-key}_i,\texttt{Op}(\text{info-key}_i)\\ & \texttt{Op}(\text{info-key}_{-i}),\enc_{\text{info-key}_{-i}}(\text{info}_{-i}),\texttt{Op}(\enc_{\text{info-key}_{-i}}(\text{info}_{-i})),\texttt{Com}(\text{info-key}_{-i})).
\end{align*}

In this stage, the simulator is given $Q,\pi_i(Q),\text{info}_i$ in the ideal world. Still the simulator can simulate the view by generating a random order $\pi_{-i}$ and several independent random strings, without being distinguished, according to the property of commitment scheme and global random oracle.

In the Compute Payment Function stage, the buyer and seller $i$ will additionally know the output of the payment function, which is also learned by the simulator in the ideal world. 

After the Open and Check Goods stage, the buyer and seller $i$ additionally view
$$ (\texttt{Op}(\text{info-key}_{-i}),\text{info-key}_{-i})$$ 
In the ideal world, the simulator is given $\text{info}_{-i}$. Similar with the single buyer setting, the simulator can simulate the additional view by employing the the programmability of the global random oracle and output $(\reallywidehat{\texttt{Op}(\text{info-key}_{-i})},\reallywidehat{\text{info-key}_{-i}})$ that are consistent with previously outputted $\reallywidehat{\enc_{\text{info-key}_{-i}}(\text{info}_{-i})},\reallywidehat{\texttt{Com}(\text{info-key}_{-i})}$.

\subsection{Security against buyer and two sellers: $\{\mathcal{B},\mathcal{S}_1,\mathcal{S}_2\}$-private}

This case is a trivial simulation since the simulator knows all information when adversaries are semi-honest. For completeness, we still present the simulation here. 

Before the Compute Payment Function stage, the buyer and two sellers' view in $\Pi$ is 
\begin{align*}
\view_{\mathcal{B},\mathcal{S}_1,\mathcal{S}_2}^{\Pi_{s-s}}=&(Q,\{\pi_i,\texttt{Op}(\pi_i(Q)),\text{info}_i, \texttt{Op}(\enc_{\text{info-key}_i}(\text{info}_i)),\\ &\text{info-key}_i, \texttt{Op}(\text{info-key}_i)|i=1,2\}).
\end{align*}

In this stage, $\simu_{\mathcal{B},\mathcal{S}_1,\mathcal{S}_2}$ are given $(Q,\pi_i, \text{info}_i,i=1,2)$. Then the simulator will generate other information easily as SMind described. 

In the payment computation stage, all traders will additionally view the output of the payment function, which is also viewed by the simulator in the ideal world. Thus, all traders' view can still be simulated. From now, all traders will not observe new information. Thus, the simulation is finished.

Thus, we have proved that $\Pi$ is secure. Based on the sequential composition technique (Lemma~\ref{lem:sq}) and the fact that MPC is secure (Lemma~\ref{lem:mpc}), we finish the formal security proof of SMind.

\end{document}